\RequirePackage{ifthen}

\newcommand{\typof}{1} %

\newcommand{\longv}[1]{\ifthenelse{\equal{\typof}{0}}{}{#1}}
\newcommand{\shortv}[1]{\ifthenelse{\equal{\typof}{0}}{#1}{}}

\newcommand{\longshortv}[2]{\ifthenelse{\equal{\typof}{0}}{#2}{#1}}
\newcommand{\drop}[0]{\ifthenelse{\equal{\typof}{0}}{}{}}

\longshortv{
\documentclass[a4paper,10pt]{article}}{
\documentclass[runningheads]{llncs}}

\usepackage{graphicx, amsmath, amssymb, mathrsfs, stmaryrd, color, adjustbox}
\usepackage[hyphens]{url}
\longv{\usepackage{a4wide}}

\newenvironment{varitemize}
{
\begin{list}{\shortv{\labelitemii}\longv{\labelitemi}}
{\setlength{\itemsep}{0pt}
 \setlength{\topsep}{0pt}
 \setlength{\parsep}{0pt}
 \setlength{\partopsep}{0pt}
 \setlength{\leftmargin}{15pt}
 \setlength{\rightmargin}{0pt}
 \setlength{\itemindent}{0pt}
 \setlength{\labelsep}{5pt}
 \setlength{\labelwidth}{10pt}
}}
{
 \end{list} 
}

\newcounter{number}

\newenvironment{varenumerate}
{\begin{list}{\arabic{number}.}
  {
   \usecounter{number}
   \setlength{\labelwidth}{4.0mm}
   \setlength{\labelsep}{2.0mm}
   \setlength{\itemindent}{0.0mm}
   \setlength{\itemsep}{0.0mm}
   \setlength{\topsep}{0.0mm}
   \setlength{\parskip}{0.0mm}
   \setlength{\parsep}{0.0mm}
   \setlength{\partopsep}{0.0mm}
  }
}
{\end{list}}

\newcommand{\PA}{\mathsf{PA}}
\newcommand{\HA}{\mathsf{HA}}

\newcommand{\Nat}{\ensuremath{\mathbb{N}}}    
\newcommand{\Bool}{\ensuremath{\mathbb{B}}}

\newcommand{\seq}{\omega}

\newcommand{\Cyl}[1]{\mathsf{C}_{#1}} 
\newcommand{\cylOp}{\cdot}

\newcommand{\cylS}[1]{\mathscr{C}_{#1}}
\newcommand{\muCyl}{\mu_{\cylS{}}}

\newcommand{\twoOm}{\Bool^{\Nat}}

\newcommand{\PM}{\mathscr{P}}

\newcommand{\PPA}{\mathsf{MQPA}}
\newcommand{\groundS}{\mathcal{G}}

\newcommand{\zero}{\mathtt{0}}
\newcommand{\suc}[1]{\mathtt{S}(#1)}
\newcommand{\add}{+}
\newcommand{\mul}{\times}

\newcommand{\fone}{F}
\newcommand{\ftwo}{G}
\newcommand{\fthree}{H}

\newcommand{\BOX}{\mathbf{C}}
\newcommand{\DIA}{\mathbf{D}}

\newcommand{\env}{\xi}
\newcommand{\sem}[2]{\llbracket #1\rrbracket_{#2}}

\newcommand{\IMT}{\mathtt{IMT}}
\newcommand{\RWT}{\mathtt{RW}}

\newcommand{\atm}[1]{\mathsf{FLIP}(#1)}

\newcommand{\Distr}[1]{\mathbb{D}(#1)}

\newcommand{\PR}{\mathcal{PR}}
\newcommand{\QPR}{\mathcal{OR}}

\newcommand{\TM}{\mathcal{M}}

\newcommand{\qtorand}[1]{#1^{\#}}
\newcommand{\Ext}{\textsf{EXT}}

\newcommand{\arithmetical}{arithmetical}
\newcommand{\arithmetization}{arithmetization}

\newcommand{\compby}[1]{\langle #1\rangle}

\longv{
\newtheorem{theorem}{Theorem}
\newtheorem{proposition}{Proposition}
\newtheorem{lemma}{Lemma}
\newtheorem{definition}{Definition}

\newtheorem{example}{Example}

\newenvironment{proof}{\begin{trivlist}
       \item[\hskip \labelsep {\bfseries Proof.}]}{\hfill $\Box$ \end{trivlist}}
}

\newcommand{\comp}{\odot}
\begin{document}
\shortv{
\title{On Measure Quantifiers\\ in First-Order Arithmetic\thanks{Supported by ERC CoG ``DIAPASoN'', GA 818616.}}}

\longv{\title{On Measure Quantifiers in First-Order
Arithmetic\thanks{Supported by ERC CoG ``DIAPASoN'', GA 818616.} \\ (Long Version)}}
\longshortv{
  \author{Melissa Antonelli \and Ugo Dal Lago \and Paolo Pistone}
  \date{}
}{
\titlerunning{On Measure Quantifiers in First-Order Arithmetic}
%
\author{Melissa Antonelli
\and
Ugo Dal Lago
\and
Paolo Pistone 
}

\authorrunning{M. Antonelli et al.}
%
\institute{University of Bologna, Bologna, Italy \\
\email{\{melissa.antonelli2,ugo.dallago,paolo.pistone2\}@unibo.it}
}
}

\maketitle        
\begin{abstract}
  We study the logic obtained by endowing the language of first-order
  arithmetic with second-order measure quantifiers.  This new kind of
  quantification
  allows us to express that the argument formula is true \emph{in
    a certain portion} of all possible interpretations of the
  quantified variable.  We show that first-order arithmetic with
  measure quantifiers is capable of 
  formalizing simple results from
  probability theory and, most importantly, of
  representing every
  recursive random function. Moreover, we introduce a realizability
  interpretation of this logic in which programs have access to an
  oracle from the Cantor space.
\end{abstract}

\shortv{\keywords{Probabilistic Computation \and Peano Arithmetic \and Realizability}}

\section{Introduction}

The interactions between first-order arithmetic and the theory of computation  
are plentiful and deep.
On the one side, proof systems for arithmetic can be used to
prove termination of certain classes of algorithms~\cite{SorUrz}, or
to establish complexity bounds~\cite{Buss86}.
On the other,
higher-order programming languages, such as typed $\lambda$-calculi,
can be proved to capture the computational content of arithmetical
proofs.  These insights can be pushed further, giving rise to logical
and type theories of various strengths.  
Remarkably, all the quoted research directions rely
on the tight connection between the concepts of \emph{totality} (of
functions) and \emph{termination} (of algorithms). 

However, there is one side of the theory of computation which
was only marginally touched by this fruitful interaction, that is, randomized
computation. 
Probabilistic models 
have been widely investigated and
are nowadays pervasive in many
areas of computer science.
The idea of relaxing the notion of
algorithm to account for computations involving random decisions
appeared early in the history of modern computability theory and studies on probabilistic computation have been developed
since the 1950s and 1960s\shortv{~\cite{LMSS}}\longv{~\cite{Carlyle,LMSS,Davis61,Santos68,Santos69b,Simon81}}.
Today several formal models
are available, such as probabilistic automata\shortv{~\cite{Rabin63}}\longv{~\cite{Segala95,Rabin63}}, \longv{both Markovian and oracle }probabilistic Turing machines (from now on,
PTMs)\shortv{~\cite{Santos69,Gill74}}\longv{~\cite{Davis61,Santos69,Gill74,Gill77}}, and
probabilistic
$\lambda$-calculi\shortv{~\cite{SahebDjaromi}}\longv{~\cite{SahebDjaromi,JonesPlotkin}. 
At this point randomized computation is ubiquitous}. 

In
probabilistic computation, behavioral properties, such as
termination, have a \emph{quantitative} nature: 
any computation terminates \emph{with a given probability}. Can such quantitative properties be studied within a logical system? 
Of course, logical systems for set-theory and second-order logic can be expressive enough 
to represent measure theory~\cite{Simpson} and, thus, are inherently capable of talking about randomized
computations. 
Yet, what should one add to \emph{first-order} arithmetic to make it capable of describing probabilistic computation?

In this paper we provide an answer to this question by introducing a 
%
%
%
 somehow \emph{minimal} extension of
first-order Peano Arithmetic by means of measure quantifiers. We will call this system $\PPA$. 
Its language is obtained by enriching the language of $\PA$ with a special unary predicate, $\atm \cdot$, whose interpretation is an element of the Cantor space, 
$\{0,1\}^\Nat$, 
and with measure-quantified formulas, such as $\BOX^{\frac{1}{2}}\fone$, which expresses the fact that $\fone$ has probability $\geq \frac{1}{2}$ of being true (that is, the subset of 
$\{0,1\}^\Nat$ 
which makes $A$ true has measure $\geq \frac{1}{2}$).
The appeal to the Cantor space is essential here, since
there is no \emph{a priori} bound on the amount of random bits a given 
computation might need; at the same time, we show that it yields a very natural measure-theoretic semantics.

The rest of this paper is structured as follows.  In Section~\ref{Section1} we introduce the
syntax and semantics of $\PPA$.
Section~\ref{monkey} shows that some non-trivial results in probability theory can be naturally expressed in $\PPA$. In Section~\ref{Arithmetization} 
we establish our main result, that is, a representation theorem within $\PPA$ for random functions computed by PTMs, which is the 
probabilistic analogous to G\"odel's arithmetization theorem for recursive functions in $\PA$~\cite{Godel}. Finally, in Section~\ref{Realizability}, a realizability interpretation for $\PPA$ in terms of computable functions with oracles on the Cantor space is presented.
\shortv{For further details, 
an extended version of this work is available, see~\cite{longVersion}.}

\section{Measure-Quantified Peano Arithmetic}\label{Section1}

This section is devoted to the introduction 
of the syntax and semantics for formulas of $\PPA$. 
Before the actual presentation, 
we need some (very modest) preliminaries
from measure theory.

\paragraph{Preliminaries.}
The standard model $(\Nat, +, \times)$ has nothing probabilistic in itself. 
Nevertheless, it can be naturally extended into a probability space:
 arithmetic being discrete, one may consider the underlying
sample space as just $\Bool^{\Nat}$, namely the set of all infinite
sequences of elements from $\Bool$ = $\{0,1\}$. We will use
metavariables, such as $\omega_{1}, \omega_{2}, \dots$, for the elements
of $\Bool^{\Nat}$.  
As it is known,
there are standard ways of building a well behaved
$\sigma$-algebra and a probability space on $\Bool^{\Nat}$, which we will briefly recall
here.  
The subsets of $\Bool^{\Nat}$ of the form
	$$
	\Cyl{X} = \{s \cylOp \seq \ | \ s \in X \ \& \ \seq \in \Bool^{\Nat}\},
	$$
where $X \subseteq \Bool^{n}$ and $\cdot$
denotes sequence concatenation,
are called \emph{$n$-cylinders}~\cite{Billingsley}.
Specifically, we are interested in $X$s defined
as follows:
$X_{n}^{b}=\{s\cdot b \ | \ s \in \Bool^{n} \ \&  \ b \in \Bool\} \subseteq \Bool^{n+1}$, with $n\in \Nat$. 
We will often deal with cylinders of the form
$\Cyl{X^{1}_{n}}$.
We let $\cylS{n}$ and $\cylS{}$ indicate the set of all $n$-cylinders and the corresponding algebra, made of
the open sets of the natural topology on
$\Bool^{\Nat}$. 
The smallest $\sigma$-algebra
including $\cylS{}$, which is Borel, is indicated as $\sigma(\cylS{})$.  
There is a natural way of defining a probability measure $\muCyl$ on $\cylS{}$, 
namely by assigning to $\Cyl{X}$ the measure $\frac{|X|}{2^{n}}$. 
There exists canonical ways to extend this
to $\sigma(\cylS{})$.
In doing so, the standard model ($\Nat, +, \times$) 
can be generalized to
$\PM$ = $(\Nat, +, \times, \sigma(\cylS{}), \muCyl)$, which will be our standard model for $\PPA$.\footnote{Here, we will focus on this structure as a ``standard model'' of $\PPA$, leaving 
the study of alternative models for future work.} 
	When interpreting sequences in $\Bool^{\Nat}$ as infinite
supplies of random bits, the set of sequences such that the $k$-th coin flip's 
result is $1$ (for any fixed $k$) is assigned measure $\frac{1}{2}$,
meaning that each random bit is uniformly distributed and independent from the others.
\paragraph{Syntax.}
We now introduce the syntax of $\PPA$. 
Terms are defined as in classic first-order arithmetic. 
Instead, the formulas of $\PPA$ are obtained by endowing
the language of $\PA$ with \emph{flipcoin
formulas}, such as $\atm{t}$, and \emph{measure-quantified formulas},
as for example
$\BOX^{t/s}F$ and $\DIA^{t/s}F$.
Specifically, $\atm{\cdot}$ is a special unary predicate with an intuitive computational meaning.
It basically provides an infinite supply of independently
and randomly distributed bits. 
Intuitively, given a closed term $t$,
$\atm{t}$ 
holds if and only if the $n$-th tossing returns 1, 
where $n$ is the denotation of
$t+1$.
\begin{definition}[Terms and Formulas of $\PPA$]
Let $\groundS$ be a denumerable set of \emph{ground variables},
whose elements are indicated by metavariables such as $x,
y.$ The \emph{terms of $\PPA$}, denoted by $t, s$, are defined
as follows:
$$
t, s := x \mid \zero \mid \suc{t} \mid t \add s \mid t \mul s.
$$
The \emph{formulas of $\PPA$} are defined by the following grammar:
$$
\fone, \ftwo := \atm{t} \mid (t=s) \mid \neg \fone \mid \fone \vee \ftwo \mid \fone \wedge \ftwo \ \mid \exists x.\fone \mid \forall x. \fone \mid \BOX^{t/s}\fone \mid \DIA^{t/s}\fone.
$$ 
\end{definition}
\paragraph{Semantics.}
Given an environment $\env: \groundS\to \Nat$,
the interpretation $\sem{t}{\env}$ of a term $t$ is defined as usual.
%
%
\longv{
\begin{definition}[Semantics for Terms of $\PPA$]
An \emph{environment} $\env$ is a mapping that assigns to each ground
variable a natural number, $\env : \groundS \rightarrow \Nat$. Given a
term $t$ and an environment $\env$, the \emph{interpretation of $t$ in
$\env$} is the natural number $\sem{t}{\env} \in \Nat$, inductively
defined as follows:
\par
\begin{minipage}{\linewidth}
\begin{minipage}[t]{0.4\linewidth}
\begin{align*}
\sem{x}{\env} &:= \env(x) \in \Nat \\
\sem{\zero}{\env} &:= 0 \\
\sem{\suc{t}}{\env} &:= \sem{t}{\env} + 1
\end{align*}
\end{minipage}
\hfill
\begin{minipage}[t]{0.55\linewidth}
\begin{align*}
\sem{t \add s}{\env} &:= \sem{t}{\env} + \sem{s}{\env} \\
\sem{t \mul s}{\env} &:= \sem{t}{\env} \times \sem{s}{\env}
\end{align*}
\end{minipage}
\end{minipage}
\normalsize
\end{definition}}
%
%
Instead, the interpretation of formulas requires a little
care, being it inherently \emph{quantitative}: any formula $\fone$ is associated with a \emph{measurable} set, $\sem{\fone}{\env}\in \sigma(\cylS{})$ (similarly to e.g.~\cite{MSM}). 
\begin{definition}[Semantics for Formulas of $\PPA$]\label{Semantics}
Given a formula $\fone$ and an environment $\env$,
the \emph{interpretation of $\fone$ in $\env$} is
the {measurable} set of sequences
$\sem{\fone}{\env} \in \sigma(\cylS{})$ inductively defined as follows:
\par
\begin{minipage}{\linewidth}
\begin{minipage}[t]{0.4\linewidth}
\begin{align*}
\sem{\atm{t}}{\env} &:= \Cyl{X^{1}_{\sem{t}{\env}}} \\
\sem{t = s}{\env} &:= \begin{cases} \twoOm \ \ &\text{ if } \sem{t}{\env} = \sem{s}{\env} \\
\emptyset \ \ &\text{ otherwise} \end{cases} \\
\sem{\neg \ftwo}{\env} &:= \twoOm \ – \ \sem{\ftwo}{\env} 
\end{align*}
\end{minipage}
\hspace{3mm}
\begin{minipage}[t]{0.55\linewidth}
\begin{align*}
\sem{\ftwo \vee \fthree}{\env} &:= \sem{\ftwo}{\env} \cup \sem{\fthree}{\env} \\
\sem{\ftwo \wedge \fthree}{\env} &:= \sem{\ftwo}{\env} \cap \sem{\fthree}{\env} \\
\sem{\exists x.\ftwo}{\env} &:= \bigcup_{i \in \Nat} \sem{\ftwo}{\env\{x \leftarrow i\}} \\
\sem{\forall x.\ftwo}{\env} &:= \bigcap_{i \in \Nat}\sem{\ftwo}{\env\{x \leftarrow i\}}
\end{align*}
\end{minipage}
\end{minipage}

\begin{align*}
\sem{\BOX^{t/s}\ftwo}{\env} &:= \begin{cases} \twoOm \ \ &\text{ if } \sem{s}{\env}> 0 \text{ and } \muCyl(\sem{\ftwo}{\env}) \ge \sem{t}{\env}/{\sem{s}{\env}} \\ \emptyset \ \ &\text{ otherwise} 
\end{cases} \\
\sem{\DIA^{t/s}\ftwo}{\env} &:= \begin{cases} \twoOm \ \ &\text{ if } \sem{s}{\env}= 0 \text{ or } \muCyl({\sem{\ftwo}{\env}}) < \sem{t}{\env}/{\sem{s}{\env}}\\
\emptyset \ \ &\text{ otherwise} \end{cases}
\end{align*}
\end{definition}
The semantics is well-defined since the sets $\sem{\atm{t}}{\env}$ and
$\sem{t=s}{\env}$ are measurable, and measurability is preserved
by all the logical operators.
It is not difficult to see that any $n$-cylinder can be captured as the interpretation of some $\PPA$ formula.
However, the language of $\PPA$ allows us to express more and more complex measurable sets, as illustrated in the next sections.

The notions of validity and logical equivalence are defined
in a standard way.
\longv{
\begin{definition}}
A formula of $\PPA$, $\fone$, is \emph{valid} if and only if for every $\env$, $\sem{\fone}{\env} = \twoOm$. Two $\PPA$ formulas $\fone, \ftwo$ are \emph{logically equivalent} $\fone \equiv \ftwo$ if and only if for every $\env$, $\sem{\fone}{\env} = \sem{\ftwo}{\env}$.
\longv{
\end{definition}}
 Notably, the two measure quantifiers are inter-definable, since one has 
 $\sem{\BOX^{t/s}\fone}{\xi}= \sem{\neg\DIA^{t/s}\fone}{\xi}$.
%
\longv{
\begin{lemma}
For every formula of $\PPA$, call it $\fone$:
$$
\BOX^{t/s}\fone \equiv \neg \DIA^{t/s}\fone.
$$
\end{lemma}
\begin{proof}
The proof is based on Definition \ref{Semantics},
\begin{align*}
\sem{\neg \DIA^{t/s} \fone}{\env} &= \twoOm \ - \sem{\DIA^{t/s} \fone}{\env} \\
&= \twoOm - \begin{cases}
\twoOm \ \ \ &\text{ if } \muCyl(\sem{\fone}{\env}) < \sem{t}{\env}/\sem{s}{\env} \\
\emptyset \ \ \ &\text{ otherwise}
\end{cases} \\
&= \begin{cases}
\emptyset \ \ \ &\text{ if } \muCyl(\sem{\fone}{\env}) < \sem{t}{\env}/\sem{s}{\env} \\
\twoOm \ \ \ &\text{ otherwise}
\end{cases} \\
&= \sem{\BOX^{t/s} \fone}{\env}.
\end{align*}
\end{proof}
}

The following examples illustrate the use of measure-quantifiers $\BOX^{t/s}$ and $\DIA^{t/s}$ and, in particular, the role of probabilities of the form $\frac{t}{s}$.
\begin{example}
The formula $\fone = \BOX^{1/1}\exists x.\atm x$ states that a true random bit will almost surely be met. It is valid, as the set of constantly 0 sequences forms a singleton, which has measure 0.
\end{example}

\begin{example}
The formula\footnote{For the sake of readability, $\fone$ has been written with a little abuse of notation the actual $\PPA$ formula being $\forall x.\BOX^{1/z}( \mathrm{EXP}(z,x) \land \forall{y}.( \exists w.(y+w=x)\to    \atm{y})$, where $\mathrm{EXP}(z,x)$ is an arithmetical formula expressing $z=2^{x}$ and 
$\exists w.y+w=x$ expresses $y\leq x$.
}
 $\fone = \forall x.\BOX^{1/2^{x}}\forall_{y \leq x}.\atm{y}$
  states that the probability for the first $x$ random bits to be true is at least $\frac{1}{2^{x}}$. This formula is valid too.
\end{example}

\section{On the Expressive Power of $\PPA$}\label{monkey}

As anticipated, the language of $\PPA$ allows us
to express some elementary results from probability theory, and to check their validity in the structure $\PM$. In this section we sketch a couple of examples.

\paragraph{The Infinite Monkey Theorem.}
Our first example is the so-called \emph{infinite monkey theorem} ($\IMT$). 
It is a classic result from probability theory stating that a monkey randomly typing on a keyboard has probability 1 of ending up writing the \emph{Macbeth} (or any other fixed string), sooner or later.
Let the formulas $\fone(x,y)$ and $\ftwo(x,y)$ of $\PA$ express, respectively, that ``$y$ is strictly smaller than the length of (the binary sequence coded by) $x$'', and that ``the $y$+1-th bit of $x$ is 1''. We can formalize $\IMT$ through the following formula:
$$
\fone_{\IMT} : \forall x.\BOX^{1/1}\forall y.\exists z.\forall w.\fone(x, w) \rightarrow ( \ftwo(x,w) \leftrightarrow \atm{y+z+w}).
$$
Indeed, let $x$ be a binary encoding of the \emph{Macbeth}. The formula $\fone_{\IMT}$ says then that for all choice of start time $y$, there exists a time $y+z$ after which $\atm \cdot$ will evolve exactly like $x$ with probability 1.

How can we justify $\fone_{\IMT}$ using the semantics of $\PPA$? Let $\varphi(x,y,z,w)$ indicate the formula 
$\fone(x, w) \rightarrow ( \ftwo(x,w) \leftrightarrow \atm{y+z+w})$. We must show that for all natural number $n\in \mathbb N$, there exists a measurable set $S^{n} \subseteq \Bool^{\Nat}$ of measure 1 such that any sequence in $S^{n}$ satisfies 
the formula $\forall y.\exists z.\forall w.\varphi(n,y,z,w)$.
To prove this fact, we rely on a well-known result from measure theory, namely the \emph{second Borel-Cantelli Lemma}:
\begin{theorem}[\cite{Billingsley}, Thm.~4.4, p.~55]
\label{lemma:borelcantelli}
If $(U_{y})_{y\in \mathbb N}$ is a sequence of independent events in $\Bool^{\Nat}$, and $\sum^{\infty}_{y}\muCyl(U_{y})$ diverges, then  
$\muCyl\left ( \bigcap_{y}\bigcup_{z>y}U_{z}\right)=1$.
\end{theorem}
Let us fix $n\in \mathbb N$ and let $\ell(n)$ indicate the length of the binary string encoded by $n$.
 We suppose for simplicity that $\ell(n)>0$ (as the case $\ell(n)=0$ is trivial). 
We construct $S^{n}$ in a few steps as follows:
\begin{varitemize}
\item for all $p\in \mathbb N$, let $U_{p}^{n}$ be the cylinder of sequences which, after $p$ steps, agree with $n$; observe that the sequences in $ U_{p}^{n}$ satisfy the formula $\forall w.\varphi(n,p,0,w)$;
\item for all $p\in \mathbb N$, let $V_{p}^{n}=U^{n}_{ p\cdot \ell(n)+1}$; observe that the sets $V^{n}_{p}$ are pairwise independent and $\muCyl(\sum_{p}^{\infty}V^{n}_{p})=\infty$;
%
%

\item for all $p\in \mathbb N$, let $S^{n}_{p}=\bigcup\{ U^{n}_{p+q}\mid \exists _{s>p}. p+q= s\cdot \ell(n)+1\}$. 
Observe that any sequence in $ S^{n}_{p}$ satisfies $\exists z.\forall w.\varphi(n,p,z,w)$;
Moreover, one can check that $S^{n}_{p}=\bigcup_{q>p}V^{n}_{q}$;
\item we finally let $ S^{n}:=\bigcap_{p}S^{n}_{p}$.
\end{varitemize}
We now have that any sequence in $S^{n}$ satisfies $\forall y.\exists z.\forall w.\varphi(n,y,z,w)$; furthermore, by Theorem  \ref{lemma:borelcantelli}, $\muCyl(S^{n})=\muCyl(\bigcap_{p}\bigcup_{q>p}V^{n}_{q})=1$. 
Thus, for each choice of $n\in \mathbb N$,  
$\muCyl$$(\sem{\forall y.\exists z.\forall w.\varphi(x,p,z,w)}{\{x\leftarrow n\}})$ $\ge$ 
$\muCyl(S^{n}) \ge 1$, 
and we conclude that $\sem{F_{\IMT}}{\env}=\Bool^{\Nat}$.

\paragraph{The Random Walk Theorem.}
A second example we consider is the \emph{random walk theorem} ($\RWT$): any simple random walk over $\mathbb{Z}$ starting from 1 will pass through 1 infinitely many times with probability 1. 
More formally, any $\omega \in \Bool^{\Nat}$ induces a simple random walk starting from 1, by letting the $n$-th move be right if $\omega(n)=1$ holds and left if $\omega(n)=0$ holds. One has then:

\begin{theorem}[\cite{Billingsley}, Thm.~8.3, p.~117]\label{lemma:borelcantelli2}
Let $U_{ij}^{(n)}\subseteq \Bool^{\Nat}$ be the set of sequences for which the simple random walk starting from $i$ leads to $j$ in $n$ steps. Then $\muCyl\left (\bigcap_{x}\bigcup_{y\geq x} U_{11}^{(y)}\right)=1$.
\end{theorem}
Similarly, the random predicate $\atm n$ induces a simple random walk starting from 1, by letting the $n$-th move be right if $\atm n$ holds and left if $\neg \atm n$ holds.  To formalize $\RWT$ in $\PPA$ we make use two arithmetical formulas:
\begin{varitemize}
\itemsep0em
 \item $H(y,z)$ expresses that $y$ is even and $z$ is the code of a sequence of length $\frac{y}{2}$, such that for all $i,j < \frac{y}{2}$, $z_{i} < y$,
 and $z_{i} = z_{j} \Rightarrow i = j$ 
(that is, $z$ codes a subset of
 $\{0, \dots, y–1\}$ of cardinality
 $\frac{y}{2}$);
 \item $K(y,z,v) =H(y,z)\land \exists i.i < \frac{y}{2} \wedge z_{i}=v$.
\end{varitemize}
The formula of $\PPA$ expressing $\RWT$ is as follows: 
$$
F_{\RWT} : \BOX^{1/1}\forall x.\exists y.\exists z. y\ge x  \wedge  H(y,z)  \wedge  \forall v.\Big(v<y \rightarrow \big(K(y,z,v) \leftrightarrow \atm v\big)\Big).
$$
\normalsize
$F_{\RWT}$ basically says that for any fixed $x$, we can find $y \ge x$ and a subset $z$ of
$\{0, \dots, y-1\}$ of cardinality
$\frac{y}{2}$, containing all and only the values $v<y$ such that $\atm v$ holds (so that the number of $v<y$ such that $\atm v$ holds coincides with the number of $v < y$ such that $\neg \atm v$ holds). This is the case precisely when the simple random walk goes back to 1 after exactly $y $ steps.

To show the validity of $F_{\RWT}$ we can use the measurable set $S=\bigcap_{n}\bigcup_{p\geq n}U_{11}^{(p)}$. Let $\psi(y,z,v)$ be the formula $(v<y \rightarrow (K(y,z,v) \leftrightarrow \atm v))$. 
Observe that any sequence in $U_{11}^{(n)}$ satisfies the formula $\exists z. H(n,z)\land \forall v.\psi(y,z,v,w)$.
Then, any sequence in $S$ satisfies the formula
$\forall x.\exists y.\exists z. y\geq x\land  H(y,z)\land \forall v.\psi(y,z,v)$. Since, by Theorem \ref{lemma:borelcantelli2}, $\muCyl(S)=1$, we conclude that 
$\muCyl(\sem{F_{\RWT}}{\env})\geq \muCyl(S)\geq 1$, and thus that $\sem{F_{\RWT}}{\env}=\Bool^{\Nat}$.

\section{Arithmetization}\label{Arithmetization}

It is a classical result in computability theory\shortv{~\cite{Godel,SorUrz}}\longv{~\cite{Godel,Godel65,Godel92,Smith,SorUrz}} that all computable
functions are arithmetical, that is, 
for each partial recursive
function $f : \Nat^{m} \rightharpoonup \Nat$ there is an arithmetical formula 
$\fone_{f}$, such that for every $n_{1}, \dots, n_{m}, l \in \Nat$:
$f(n_{1}, \dots, n_{m}) = l$ $\Leftrightarrow$ $(\Nat, +, \times) \vDash \fone_{f}(n_{1}, \dots, n_{m}, l)$.
In this section we show that, by considering {arithmetical formulas of $\PPA$}, this fundamental result can be generalized to computable \emph{random} functions.

\paragraph{Computability in Presence of Probabilistic Choice.}
Although standard computational models are built around determinism,
from the 1950s on, models for \emph{randomized} computation started to receive
wide attention\shortv{~\cite{LMSS,Santos69,Gill74}}\longv{~\cite{LMSS,Davis61,Rabin63,Santos69,Santos71,Gill74,Gill77,Simon81}}.
The first formal definitions 
of probabilistic Turing machines (for short, PTM) are due
to Santos\shortv{~\cite{Santos69}}\longv{~\cite{Santos69,Santos71}} and Gill~\cite{Gill77,Gill74}.
Roughly, a PTM is an ordinary Turing machine (for short, TM) with the additional capability
of making random decisions.
\longv{Two alternative paradigms have been developed in the literature, the so-called \emph{Markovian} and \emph{oracle} PTMs.}
Here, we consider the definition by Gill, in which 
the probabilistic choices performed by the machines are binary and
fair. 
\begin{definition}[Probabilistic Turing Machine~\cite{Gill74,Gill77}]
 A (one-tape) \emph{probabilistic Turing machine} is a 5-tuple $(Q, \Sigma, \delta, q_{0}, Q_{f})$, whose elements are defined as in a standard TM, except for the
 probabilistic transition function $\delta$,
 which, given the current (non-final) state and symbol, specifies two equally-likely transition steps.
\end{definition}
As any ordinary TM computes a partial function
on natural numbers, PTM can be seen as computing a
so-called \emph{random function}~\cite[pp. 706--707]{Santos69}.
Let $\Distr{\Nat}$ indicate the set of \emph{pseudo-distributions} on $\Nat$, i.e.~of functions $f:\Nat \to \mathbb R_{[0,1]}$, such that $\sum_{n\in \Nat}f(n)\leq 1$. Given a PTM, $\TM$,
a random function is a function $\compby{\TM}:\Nat\rightarrow\Distr{\Nat}$ which, for each natural
number $n$, returns a pseudo-distribution supporting all
the possible outcomes $\TM$
produces when fed with (an encoding of) $n$ in input,
each with its own probability.\longv{\footnote{The seminal definition of \emph{random function} 
already appeared in~\cite{Santos69}:
\begin{quote}
\textsc{Definition.} A $k$-ary random function $\phi$ is a function from $E^{n+1}$, the collection of all $(k+1)$-tuples of nonnegative integers, to [0,1] satisfying
$$
\sum_{m=0}^{\infty}\phi(m_{1}, m_{2}, \cdots, m_{k}, m) \leq 1
$$
for every $k$-tuple ($m_{1}, m_{2}, \dots, m_{k})$. \cite[pp. 706-707]{Santos69}
\end{quote}
Remarkably, in the same paper, Santos also delineated the notion of \emph{probabilistic transition function}, on which the Markovian paradigm of PTM is based~\cite[p. 705]{Santos69}~\cite[p. 170]{Santos71}.}} As expected, the random function $f:\Nat\rightarrow\Distr{\Nat}$
is said to be \emph{computable} 
when there is a PTM, $\TM$,
such that $\compby{\TM}=f$.

\longv{
Another widespread 
definition of TM is the one based on the
notion of oracle.
An \emph{oracle TM} 
is a pair consisting of a standard deterministic TM
and a random-bit oracle 
(for further details, see for example~\cite{DK}).  
The random-bit oracle takes the form of an \emph{oracle tape}, which is consulted whenever a coin-tossing state is encountered.
The two definitions are assumed to be equivalent but, to the best of the authors' knowledge, no formal proof
of this equivalence has been presented in the literature yet.
}

\paragraph{Stating the Main Result.}
In order to generalize G\"odel's arithmetization of partial recursive functions to the
class of computable random functions, we \shortv{first }\longv{start by}
introduc\shortv{e}\longv{ing} 
the notion of
arithmetical random function.

\begin{definition}[Arithmetical Random Function]\label{arithmetical}
A random function $f : \Nat^{m}\rightarrow\Distr{\Nat}$
is said to be \emph{\arithmetical} if and only if 
there is a formula of $\PPA$,
call it $F_{f}$, with free variables $x_{1}, \dots, x_{m},y$, such that for every
$n_{1}, \dots, n_{m},l\in \Nat$, it holds that:
\begin{align}
\muCyl\big(\sem{F_f(n_{1}, \dots, n_{m},l)}{}\big)=f(n_{1}, \dots, n_{m})(l).
\end{align}
\end{definition}
The \emph{\arithmetization~theorem} below relates random functions and the formulas of $\PPA$, and is the main result of this paper.
\begin{theorem}\label{theorem:arithmetization}
All computable random functions are \arithmetical.
\end{theorem}
Actually, we establish a stronger fact. Let us call a formula $A$ of $\PPA$ $\Sigma_{1}^{0}$ if $A$ is equivalent to a formula of the form $\exists x_{1}.\dots.\exists x_{n}.A'$, where $A'$ contains neither first-order nor measure quantifiers. 
Then Theorem~\ref{theorem:arithmetization} can be strengthened by saying that any computable random function is represented (in the sense of Definition~\ref{arithmetical}) by a $\Sigma_{1}^{0}$-formula of $\PPA$.

Moreover, we are confident that a sort of converse of this fact can be established, namely that for any  $\Sigma_{1}^{0}$-formula 
$A(x_{1},\dots, x_{m})$ there exists a computable random relation $r(x_{1},\dots, x_{m})$ (i.e.~a computable random function such that $r(x_{1},\dots,x_{n})(i)=0$ whenever $i\neq 0,1$) such that 
$\muCyl(\sem{A(n_{1},\dots,n_{m})}{})=r(n_{1},\dots, n_{m})(0)$ and 
$\muCyl(\sem{\lnot A(n_{1},\dots,n_{m})}{})=r(n_{1},\dots, n_{m})(1)$. However, we leave this fact and, more generally, the exploration of an \emph{arithmetical hierarchy} of randomized sets and relations, to future work.

Given the conceptual distance existing between TMs and arithmetic, a direct proof of Theorem~\ref{theorem:arithmetization} would be
cumbersome. 
It is thus convenient to look
for an alternative route.

\paragraph{On Function Algebras.}
In\shortv{~\cite{GDLZ}}\longv{~\cite{GDLZ,DalLagoZuppiroli}}, the class $\PR$ of \emph{probabilistic} or \emph{random recursive functions} is defined as a generalization of 
Church and Kleene's standard
one\shortv{~\cite{ChurchKleene,Kleene36}}\longv{~\cite{ChurchKleene,Kleene36,Kleene36c,Kleene43}}. 
$\PR$ is characterized as
the smallest class of functions, which
(i) contains some {basic random functions}, and (ii) is
closed under composition, primitive recursion and minimization.
For
all this to make sense, composition and primitive recursion
are defined following the \emph{monadic} structure of $\Distr{\cdot}$. 
In order to give 
a presentation as straightforward as possible, we preliminarily introduce the notion of \emph{Kleisli extension} of a function with values in $\Distr{\Nat}$.
\begin{definition}[Kleisli Extension]
\longv{Given a function $f : \Nat \rightarrow \Distr{\Nat}$,
its \emph{(simple) Kleisli extension} $f^{\mathbf{K}}
: \Distr{\Nat} \rightarrow \Distr{\Nat}$ is defined as follows:
$$
f^{\mathbf{K}}(d)(n) = \sum_{i\in \Nat}d(i) \cdot f(i)(n).
$$ 
More in general, g}\shortv{G}iven a $k$-ary function $f :
X_{1} \times \dots \times X_{i-1}\times \Nat \times
X_{i+1}\times \dots \times X_{k} \rightarrow \Distr{\Nat}$, its
\emph{$i$-th Kleisli extension} $f^{\mathbf{K}}_i : X_{1} \times \dots\times X_{i-1}\times \Distr{\Nat}\times X_{i+1}\times\dots \times X_{k}
\rightarrow \Distr{\Nat}$ is defined as follows:
$$
f^{\mathbf{K}}_i(x_{1}, \dots, x_{i-1}, d, x_{i+1}, \dots, x_{k})(n)
= \sum_{j \in \Nat}d(j) \cdot f(x_{1}, \dots, x_{i-1},j,x_{i+1},\dots, x_{k})(n).
$$
\end{definition}
The construction at the basis of the $\mathbf{K}$-extension\longv{'s general case} can be applied more than once.\longv{\footnote{For example, let us consider a binary function $f:\Nat\times\Nat\rightarrow\Distr{\Nat}$,
its \emph{total} $\mathbf{K}$-extension is as follows:
\begin{align*}
f_{1}^{\mathbf{K}}(d_{1}, d_{2})(y) &= \sum_{i_{1} \in \Nat}d_{1}(i_{1}) \cdot f^{\mathbf{K}}_2(i_{1}, d_{2})(y) \\
&= \sum_{i_{1} \in \Nat}d_{1}(i_{1}) \cdot \Big(\sum_{i_{2} \in \Nat}d_{2}(i_{2}) \cdot f(i_{1}, i_{2})(y)\Big) \\
&= \sum_{i_{1},i_{2} \in \Nat}f(i_{1}, i_{2})(y) \cdot d_{1}(i_{1}) \cdot d_{2}(i_{2}) \\
&= \sum_{i_{1}, i_{2} \in \Nat}f(i_{1},i_{2})(y) \cdot \prod_{k \in \{1,2\}}d_{k}(i_{k}).
\end{align*}}}
Specifically, given a function $f : \Nat^{k} \rightarrow \Distr{\Nat}$, its \emph{total} $\mathbf{K}$-extension $f^{\mathbf{K}} : \left(\Distr{\Nat}\right)^k \rightarrow \Distr{\Nat}$ is defined as follows:
$$ 
f^{\mathbf{K}}(d_{1}, \dots, d_{k})(n) = \sum_{i_{1}, \dots, i_{k} \in \Nat}f(i_{1}, \dots, i_{k})(n) \cdot \prod_{1\leq j\leq k} d_{j}(i_{j}).
$$
We can now define the class $\PR$ formally as follows:
\begin{definition}[The Class $\PR$~\cite{GDLZ}]
The \emph{class of probabilistic recursive functions}, $\PR$, is the smallest class of probabilistic functions containing:
\begin{varitemize}

\item The \emph{zero function}\longv{$, z : \Nat \rightarrow \Distr{\Nat}$, such that }\shortv{ : }for every $x \in \Nat$, $z(x)(0) = 1$;
\item The \emph{successor function}\longv{, $s : \Nat \rightarrow \Distr{\Nat}$, such that}\shortv{:} for every $x \in \Nat$, $s(x)(x+1) = 1$;
\item The \emph{projection function}\longv{, $\pi_{m}^{n} : \Nat^{n} \rightarrow \Distr{\Nat}$, defined as }\shortv{: }for $1\leq m\leq n$, $\pi_{m}^{n}(x_{1}, \dots, x_{n})(x_{m}) = 1$;
\item The \emph{fair coin function}, $r : \Nat \rightarrow \Distr{\Nat}$, defined as follows:
$$
r(x)(y) = \begin{cases}
\frac{1}{2} \ \ \ &\text{ if  } y = x \\
\frac{1}{2} \ \ \ &\text{ if  } y = x + 1\shortv{;} \longv{\\
0 \ \ \ &\text{ otherwise;}}
\end{cases}
$$
\end{varitemize}
and closed under:
\begin{varitemize}
\itemsep0em
\item
  \emph{Probabilistic composition}. Given $f : \Nat^{n} \rightarrow \Distr{\Nat}$ and $g_{1}, \dots, g_{n} : \Nat^{k} \rightarrow \Distr{\Nat}$, their composition is \longv{a function $f \comp (g_{1}, \dots, g_{n}) : \Nat^{k} \rightarrow \Distr{\Nat}$ }defined as follows:\longv{\footnote{That is,
\begin{align*}
((f\comp (g_{1}, \dots, g_{n}))(\mathtt{x}))(y) &= (f^{\mathbf{K}}(g_{1}(\mathtt{x}), \dots, g_{n}(\mathtt{x})))(y) \\
&= \sum_{i_{1}, \dots, i_{n}}f(i_{1}, \dots, i_{k})(y) \cdot \prod_{1 \leq j \leq n}g_{j}(\mathtt{x})(i_{j}).
\end{align*}
The simplest case, is teat of unary composition which is defined as follows, given $f : \Nat \rightarrow \Distr{\Nat}$ and $g : \Nat \rightarrow \Distr{\Nat}$, the function $h : \Nat \rightarrow \Distr{\Nat}$ obtained by composition from $f$ and $g$ is:
$$
(f\comp g)(x)(y) = f^{\mathbf{K}}(g(x))(y) 
$$
Otherwise said,
$$
((f\comp g)(x))(y) = \sum_{z\in \Nat}g(x)(z)\cdot f(z)(y).
$$}}
$$
(f \comp (g_{1}, \dots, g_{n}))(\mathtt{x}) = f^{\mathbf{K}}(g_{1}(\mathtt{x}), \dots, g_{n}(\mathtt{x}));
$$
\item \emph{Probabilistic primitive recursion.} Given $f : \Nat^{k} \rightarrow \Distr{\Nat}$, and $g : \Nat^{k+2} \rightarrow \Distr{\Nat},$ the function \longv{$h : \Nat^{k+1} \rightarrow \Distr{\Nat}$ }obtained from them by primitive recursion is as follows:
\longv{
\begin{align*}
h(\mathtt{x},0) &= f(x) \\
h(\mathtt{x}, y+1) &= g^{\mathbf{K}}_{k+2}(\mathtt{x}, y, h(\mathtt{x}, y));
\end{align*}}
\shortv{
$$
h(\mathtt{x},0) = f(x) \quad \quad \quad h(\mathtt{x}, y+1) = g^{\mathbf{K}}_{k+2}(\mathtt{x}, y, h(\mathtt{x}, y));
$$}
\item \emph{Probabilistic minimization.} Given $f : \Nat^{k+1} \rightarrow \Distr{\Nat},$ the function \longv{$h : \Nat^{k} \rightarrow \Distr{\Nat}$, }obtained from it by minimization is defined as follows: 
$$
\mu f(\mathtt{x})(y) = f(\mathtt{x},y)(0)\cdot\Big(\prod_{z<y}\Big(\sum_{k>0}f(\mathtt{x},z)(k)\Big)\Big).
$$
\end{varitemize}
\end{definition}

\begin{proposition}[\cite{GDLZ}]\label{proposition}
$\PR$ coincides with the class of computable random functions.
\end{proposition}

The class $\PR$ is still conceptually far
from $\PPA$. 
In fact, while the latter has access to randomness
in the form of a global supply of random bits,
 the former can
fire random choices locally through a dedicated initial function.
To bridge the gap between the two, we introduce a third
characterization of computable random functions, which
is better-suited for our purposes.
\longv{
In doing so, we will define the class of 
oracle recursive functions, $\QPR$. Our definition is loosely inspired from 
\emph{oracle} TM, which can be defined as deterministic TM whose transition function can {query} a  \emph{random-bit tape}
$\omega \in \Bool^{\Nat}$.
}

%
%
The class of \emph{oracle recursive functions}, $\QPR$, is the smallest class of partial
functions of the form $f : \Nat^{m} \times \Bool^{\Nat} \rightharpoonup \Nat,$ which
(i) contains the class of \emph{oracle basic functions}, and (ii) is closed under
composition, primitive recursion, and 
minimization.
\longv{Formally,
\begin{definition}[The Class $\QPR$]\label{QPR}
The \emph{class of oracle recursive functions}, $\QPR$, is the smallest class of probabilistic functions 
containing:
\begin{varitemize}

\item The \emph{zero function}, $f_{0}$,
such that $f_{0}(x_{1}, \dots, x_{k}, \omega) = 0$;

\item The \emph{successor function}, $f_{s},$
such that $f_{s}(x,\omega) = x+1$;

\item The \emph{projection function}, $f_{\pi_{i}},$
 such that, for $1 \leq i \leq k$, $f_{\pi_{i}}(x_{1}, \dots, x_{k}, \omega) = x_{i}$;

\item The \emph{query function}, $f_{q},$ 
such that $ f_{q}(x,\omega) = \omega(x)$;
\end{varitemize} 
and closed under:

\begin{varitemize}

\item \emph{Oracle composition}. Given the oracle functions $h$ from $\Nat^{n}\times \Bool^{\Nat}$ to $\Nat$ and $g_{1}, \dots g_{n}$ (from $\Nat^{m}$), the function $f$ obtained by composition from them, is defined as follows: 
$$
f(x_{1}, \dots, x_{m}, \omega) = h(g_{1}(x_{1}, \dots, x_{m}, \omega), \dots, g_{n}(x_{1}, \dots, x_{m}, \omega), \omega);
$$

\item \emph{Oracle primitive recursion}. Given two oracle functions $h$ and $g$, from respectively $\Nat^{n}\times \Bool^{\Nat}$ and $\Nat^{n+2}\times \Bool^{\Nat}$ to $\Nat$, the function $f$, obtained by primitive recursion from them, is defined as follows:
$$
f(x,x_{1},\dots,x_{n}) = \begin{cases}
f(0, x_{1}, \dots, x_{n}, \omega) = h(x_{1}, \dots, x_{n}, \omega) \\
f(x+1, x_{1}, \dots, x_{n}, \omega) = g(f(x, x_{1}, \dots, x_{n}, \omega), x, x_{1}, \dots, x_{n}, \omega);
\end{cases}
$$

\item \emph{Oracle minimization}. Given the oracle function $g$ from $\Nat^{n+1}\times \Bool^{\Nat}$ to $\Nat$, the function $f$, obtained by minimization from $g$, is defined as follows:
$$
f(x_{1}, \dots, x_{n}, \omega) = \mu x(g(x_{1}, \dots, x_{n}, x, \omega) = 0).
$$
\end{varitemize}
\end{definition}
}
Remarkably, the only basic function
depending on $\omega$ is the query function. 
All the
closure schemes are independent from $\omega$ as well.

But in what sense do functions in $\QPR$ represent random functions?
In order to clarify the relationship between $\QPR$ and $\PR$, we associate each $\QPR$ function with a corresponding auxiliary function.
\begin{definition}[Auxiliary Function]\label{auxiliary}
  Given an oracle function $f
  : \Nat^{m} \times \Bool^{\Nat} \rightarrow \Nat$, the corresponding \emph{auxiliary
  function}, $f^{*} : \Nat^{m} \times \Nat \rightarrow \mathcal{P}(\Bool^{\Nat})$, is
  defined as follows:
  $
  f^{*}(x_{1}, \dots, x_{m}, y) = \{\omega \ | \ f(x_{1}, \dots, x_{m}, \omega) = y\}.
  $
\end{definition}
The following lemma ensures that the value of $f^{*}$ is always a \emph{measurable} set:

\begin{lemma}\label{lemma:measurability}
  For every oracle recursive function $f \in \QPR$, $f : \Nat^{m} \times \Bool^{\Nat} \rightarrow \Nat$, and natural numbers $x_{1}, \dots, x_{m}, y \in \Nat$,
   the set  $f^{*}(x_{1}, \dots, x_{m},
y)$ is measurable.
\end{lemma}
\longv{
\begin{proof}
We will show
 that, for each $f \in \QPR$, $f^{*} \in \sigma(\cylS{})$, by induction on the structure of oracle recursive functions:
\begin{varitemize} 
\item Let $f\in\QPR$ be an oracle basic function. There are four
possible cases:

\par Zero Function.
Let $f_{0}$
be the zero function, $f_{0}(x_{1}, \dots, x_{n}, \omega) = 0$. 
Then, 
$$
f_{0}^{*}(x_{1}, \dots, x_{n}, 0) = \{\omega \ | \ f_{0}(x_{1}, \dots, x_{n}, \omega) = 0\} = \twoOm.
$$
For Axioms 1 of the $\sigma$-algebra, $\twoOm \in \sigma(\cylS{})$, so $f_{0}^{*}\in \sigma(\cylS{})$.

\par Successor Function.
Let $f_{s}$
be the successor function $f_{s}(x, \omega) = x+1$.
Then, 
$$
f_{s}^{*}(x,x+1) = \{\omega \ | \ f_{s}(x,\omega) = x+1\} = \twoOm.
$$ 
As before, $f_{s}^{*} \in \sigma(\cylS{})$.

\par Projection Function.
Let $f_{\pi_{i}}$ 
be the projection function, $f_{\pi_{i}}(x_{1}, \dots, x_{n}, \omega) = x_{i}$, with $1\leq i \leq n$.
Then, 
$$
f_{\pi_{i}}^{*}(x_{1}, \dots, x_{n}, x_{i}) = \{\omega \ | \ f_{\pi_{i}}(x_{1}, \dots, x_{n}, \omega) = x_{i}\} = \twoOm.
$$
Again, $f_{\pi_{i}}^{*} \in \sigma(\cylS{})$.

\par Query Function.
Let $f_{q}$ 
be the query function, $f_{q}(x, \omega') = \omega'(x)$.
Then, 
$$
f_{q}^{*}(x,\omega') = \{\omega \ | \ f_{q}(x,\omega) = \omega'(x)\} = \{\omega \ | \ \omega(x) = 0\}
$$
for $\omega'(x) = 0$ or $f_{q}^{*}(x,\omega') = \{\omega$ $|$ $\omega(x) = 1\}$ if $\omega'(x) = 1$, in both cases $f_{q}^{*}(x,\omega')$ is a  (thin) cylinder, so $f_{q}^{*}(x, 0) \in \sigma(\cylS{})$. Therefore, $f_{q}^{*} \in \sigma(\cylS{})$.

\item Let $f\in\QPR$ be obtained by oracle composition, 
recursion or minimization from $\QPR$ functions.
Since the three cases are proved in a similar way,
let us take into account (simple) composition only.
Let $f : \Nat^{n} \times \Bool^{\Nat} \rightarrow \Nat$ be
obtained by (unary) composition from $h : \Nat \times \Bool^{\Nat} \rightarrow \Nat$ and $g : \Nat^{n} \times \Bool^{\Nat} \rightarrow \Nat$. Assume $f(x_{1}, \dots, x_{n}, \omega)$ = $h(g(x_{1}, \dots, x_{n}, \omega), \omega)$ = $v$.
 By Definition \ref{auxiliary}, $f^{*}(x_{1}, \dots, x_{n}, v)$ = 
$\bigcup_{z\in \Nat}\{\omega$ $|$ $h(z,\omega) = v\}$ $\cap$ $\{\omega$ $|$ $g(x_{1}, \dots, x_{n}, \omega) = z\}$ = $\bigcup_{z\in \Nat}h^{*}(z,v) \cap g^{*}(x_{1}, \dots, x_{n}, z)$.
 By IH, $h^{*},g^{*} \in \sigma(\cylS{})$ so, by Axiom 3 of $\sigma$-algebras, $h^{*} \cap g^{*} \in \sigma(\cylS{})$ as well. Therefore, since $f^{*}$ is a countable union of measurable sets, again for Axiom 3, it  is measurable as well, so $f^{*} \in \sigma(\cylS{})$.
\end{varitemize}
\end{proof}}
Thanks to Lemma~\ref{lemma:measurability}, we can associate any {oracle} recursive function $f:\Nat^{m} \times \Bool^{\Nat} \rightarrow \Nat$
with a random function 
$\qtorand{f}:\Nat^m\rightarrow\Distr{\Nat}$, 
defined as:
$
\qtorand{f}(x_1,\ldots,x_m)(y)=\muCyl(f^{*}(x_1,\ldots,x_m,y)).
$
This defines a close correspondence between the classes
$\PR$ and $\QPR$.

\begin{proposition}\label{prop2}
 For each $f \in \PR$, there is an oracle function $g\in \QPR$, such that
 $f=\qtorand{g}$.
 Symmetrically, for any $f\in\QPR$, there is a random function $g \in \PR$, such
 that $g = \qtorand{f}$.
\end{proposition}
\longv{
For our purpose, only the first part of Proposition~\ref{prop2} is necessary, namely that
for every $f \in \PR$, 
there is a $g \in \QPR$ 
such that $\qtorand{g} = f$.
This is established by means of a few intermediate steps.

First, some auxiliary notions are introduced.
A \emph{computable bijection} between 
$\Nat$ and $\Nat \times \Nat$ and 
the corresponding maps 
$\langle \cdot, \cdot \rangle : \Nat \times \Nat \rightarrow \Nat$
and 
$\pi_{1}, \pi_{2} : \Nat \rightarrow \Nat$ is fixed.
A \emph{tree} is defined as a subset 
$X$ of $\Bool^{*}$ (the finite set of strings) 
such that if $t \in X$ and $v \sqsubset t$
($\sqsubseteq$ being the prefix relation),
then $v \not \in X$. 
Given $t \in \Bool^{*}$, $\omega \in \Bool^{\Nat}$ is
said to be an \emph{n-extension} of $t$ 
if and only if 
$\omega= v \cdot t \cdot \omega'$,
where $|v| = n$ and
$\omega' \in \Bool^{\Nat}$.
The set of all $n$-extensions of $t$ is indicated
as $\Ext^{n}_{t}$
and is measurable.
Moreover, $\mu(\Ext^{n}_{t}) = \frac{1}{2^{|t|}}$
for every $n$ and $t$.
Given a tree $X$, every function $f : X \rightarrow \Nat$
is said to be an \emph{$X$-function}.
Thus, an oracle function 
$f : \Nat^{n+1} \times \Bool^{\Nat} \rightarrow \Nat \in \QPR$ 
returns a tree $X$ and an $X$-function $g$ on input $(m_{1}, \dots, m_{n})$ if and only if for every $k \in \Nat$, it holds that:
\begin{itemize}
\itemsep0em
\item $f(m_{1}, \dots, m_{n}, k, \omega)$ is defined if and only if 
$\omega \in \Ext^{k}_{t}$, where $t \in X$.
\item if $\omega \in \Ext^{k}_{t}$ and $t \in X$, then
$$
f(m_{1}, \dots, m_{n}, k, \omega) = q,
$$
where $\pi_{1}(q) = g(t)$ and $\pi_{2}(q) = |t|$.
\end{itemize}
It is now possible to prove the following preliminary lemma.
\begin{lemma}\label{lemma:aux}
For every $f : \Nat^{n} \rightarrow \Distr{\Nat} \in \PR$, there is an oracle recursive function 
$g : \Nat^{n+1} \times \Bool^{\Nat} \rightharpoonup\Nat \in \QPR$, 
such that for every $m_{1}, \dots, m_{n}$, $g$ returns a tree 
$X_{m_{1}, \dots, m_{n}}$ and 
an $X_{m_{1}, \dots, m_{n}}$-function
$h_{m_{1}, \dots, m_{n}}$ on input 
$m_{1}, \dots, m_{n}$, and
$$
f(m_{1}, \dots, m_{n})(y) = \sum_{h_{m_{1}, \dots, m_{n}}(t)=y}\frac{1}{2^{|t|}}.
$$
\end{lemma}

\begin{proof}
The proof is by induction on the structure of $f$ as an element of $\PR$.
\begin{varitemize}
\item Let $f\in \PR$ be a basic probabilistic function.
Then, there are four possible cases.

\par Zero Function. Let $z \in \PR$ be the zero function. Then, 
$g \in \QPR$ is an oracle function, 
so defined that on inputs $m_{1}, k, \omega$,
it returns the value $\langle 0,0\rangle$.
Indeed, $g$ returns the tree
$X_{m_{1}} = \{ \epsilon \}$ and
the $X_{m_{1}}$-function $h_{m_{1}}$ always returning
0,
as it can be easily checked. 
Moreover:
\begin{align*}
z(m_{1})(y) &= \begin{cases}
1 \ \ \ \text{ if } y = 0 \\
0 \ \ \ \text{ otherwise}
\end{cases} \\
&= \sum_{h_{m_{1}}(s) = y}\frac{1}{2^{|t|}}.
\end{align*}

\par
Successor Function. Let $s \in \PR$ be the successor function.
Then, $g \in \QPR$ is an oracle function,
so defined that
on inputs $m_{1}, k, \omega$,
it returns the value
$\langle m_{1}+1, 0\rangle$. 
So, $g$ returns the tree $X_{m_{1}} = \{\epsilon\}$
and the $X_{m_{1}}$-function 
$h_{m_{1}}$ always returning $m_{1}+1$.
Therefore:
\begin{align*}
s(m_{1})(y) &= \begin{cases}
1 \ \ \ \text{ if } y = m_{1} + 1 \\
0 \ \ \ \text{ otherwise}
\end{cases} \\
&= \sum_{h_{m_{1}}(t)=y} \frac{1}{2^{|t|}}.
\end{align*}

\par Projection Function. 
Let $\pi^{n}_{i} \in \PR$, with $1 \leq i \leq n$,
be the projection function.
Then, $g \in \QPR$ is an oracle function,
so defined that, on inputs 
$m_{1}, \dots, m_{n}, k, \omega$ it returns
the value $\langle m_{i}, 0\rangle$. 
Indeed, $g$ returns the tree $X_{m_{1}, \dots, m_{n}} = \{\epsilon\}$ and the $X_{m_{1}, \dots, m_{n}}$-function $h_{m_{1}, \dots, m_{n}}$ always returning $m_{i}$.
Moreover:
\begin{align*}
\pi^{n}_{i}(m_{1}, \dots, m_{n})(y) &= \begin{cases} 
1 \ \ \ \text{ if } y = m_{i} \\
0 \ \ \ \text{ otherwise}
\end{cases} \\
&= \sum_{h_{m_{1}, \dots, m_{n}}(t) = y} \frac{1}{2^{|t|}}.
\end{align*}

Fair Coin Function. Let $r \in \PR$ be the fair coin function.
Then, $g \in \QPR$ is an oracle function,
so defined that 
$g(m_{1}, k, \omega) = \langle l, 1\rangle$
where,
$$
l = \begin{cases}
m_{1} \ \ \ \text{ if } \omega[k] = 0 \\
m_{1+1} \ \ \ \text{ if } \omega[k] = 1.
\end{cases}
$$

\item 
If $f$ is obtained by either composition, primitive recursion or minimization, the argument is a bit more involved, as in defining the function $g$ we must take into account how the bits of the oracle accessed by $g$ are distributed in an independent way to each of the component functions. We will here only illustrate how this works in the case of composition.

Let then $f$ be obtained by composition from $f_{1}, \dots, f_{p}: \Nat^{n} \rightarrow \Distr{\Nat} $ and 
$f': \Nat^{p}\rightarrow \Distr{\Nat}$, i.e.~$f(m_{1},\dots, m_{n})(y)=\sum_{i_{1},\dots, i_{p}}f'(i_{1},\dots, i_{p})(y)\cdot \prod_{j=1}^{p}f_{j}(m_{1},\dots, m_{n})(i_{j})$.
By induction hypothesis there exist functions $g_{1},\dots, g_{p}: \Nat^{n+1}\times \Bool^{\Nat} \rightarrow \Nat$ and $g': \Nat^{p+1}\times \Bool^{\Nat}\rightarrow \Nat$ such that 

\begin{enumerate}
\item for all $m_{1},\dots, m_{n}\in \Nat$, each $g_{i}$ returns a tree $X^{i}_{m_{1},\dots, m_{n}}$ and an $X^{i}_{m_{1},\dots, m_{n}}$-function $h^{i}_{m_{1},\dots, m_{n}}$, and 
$f_{i}(m_{1},\dots, m_{n})(y)= \sum_{h^{i}_{m_{1},\dots, m_{n}}(t)=y}\frac{1}{2^{|t|}}$; 
\item for all $m_{1},\dots, m_{p}\in \Nat$, $g'$ returns a tree $Y_{m_{1},\dots, m_{n}}$ and a $Y_{m_{1},\dots, m_{p}}$-function $h'_{m_{1},\dots, m_{p}}$, and 
$f'(m_{1},\dots, m_{p})(y)= \sum_{h'_{m_{1},\dots, m_{p}}(t)=y}\frac{1}{2^{|t|}}$.
\end{enumerate}

We thus have that 
\begin{align*}
f(m_{1},\dots, m_{n})(y) &=
\sum_{i=1,\dots,i_{p}} \left ( \sum_{h'_{i_{1},\dots, i_{p}}(s)=y}\frac{1}{2^{|s|}}\right)\cdot
\left (\prod_{j=1}^{p}\sum_{h^{j}_{m_{1},\dots, m_{n}}(t)=i_{j}}\frac{1}{2^{|t|}}\right)\\
& = 
\sum_{i=1,\dots,i_{p}}
\left(
\sum_{ 
h'_{i_{1},\dots, i_{p}}(s)=y, 
h^{j}_{m_{1},\dots, m_{n}}(t_{j})=i_{j}
}
\frac{1}{2^{|s|+\sum_{j=1}^{p}|t_{j}|}}
\right)
\end{align*}

For all $m_{1},\dots, m_{n}\in \Nat$, let $X_{m_{1},\dots, m_{n}}= \{t_{1}\cdot \dots \cdot t_{p}\mid s\in t_{j}\in X^{j}_{m_{1},\dots, m_{n}}, s \in Y_{h_{\vec m}^{1}(t_{1}),\dots, h_{\vec m}^{p}(t_{p})}\}$.

Observe that any $v\in X_{m_{1},\dots, m_{n}}$ can be decomposed in a unique way as 
$v= t_{0}\cdot t_{1}\cdot \dots \cdot t_{p}$. In fact, if $t'_{0}\cdot t'_{1}\cdot \dots \cdot t'_{p}$ is any other decomposition, let $j\leq p$ be minimum such that $t_{j}\neq t'_{j}$. Then it must be either $t_{j}\sqsubset t'_{j}$ or $t'_{j}\sqsubset t_{j}$, which contradicts the fact that $X^{j}_{m_{1},\dots, m_{n}}$ and $Y_{h_{\vec m}^{1}(t_{1}),\dots, h_{\vec m}^{p}(t_{p})}$ are all trees. 

Using this fact we can show that  
$X_{m_{1},\dots, m_{n}}$ is also a tree: suppose $v=t_{0}\cdot t_{1}\cdot \dots \cdot t_{p}\in  X_{m_{1},\dots, m_{n}}$ and suppose $v'\in X_{m_{1},\dots, m_{n}}$, where $v'\sqsubset v$. 
Then $v'$ has a unique decomposition $t'_{0}\cdot \dots \cdot t'_{p}$ and one can easily show by induction on $j\leq p$ that $t'_{j} =t'_{j}$ holds. Hence it must be $v'=v$, against the assumption.

Let $h_{m_{1},\dots, m_{n}}:X_{m_{1},\dots, m_{n}}\to \Nat$ be defined by $h_{m_{1},\dots, m_{n}}(v)=y$, where 
$v$ uniquely decomposes as $s\cdot t_{1}\cdot \dots \cdot t_{p}$, $h^{j}_{\vec m}(t_{j})=i_{j}$ and $h'_{i_{1},\dots, i_{p}}(s)=y$. We can finally define: 
\begin{center}
\resizebox{\textwidth}{!}{$
g(m_{1},\dots, m_{n}, k, \omega)= 
\left \langle 
\pi_{1}\Big(g'\big ( \pi_{1}(L_{1}),\dots, \pi_{1}(L_{p}), k+\sum_{j=1}^{p }
L_{j}
, \omega\big)\Big ), 
\pi_{2}\Big(g'\big ( \pi_{1}(L_{1}),\dots, \pi_{1}(L_{p}), k+\sum_{j=1}^{p }
L_{j}
, \omega\big)\Big )+ R_{p}
  \right\rangle
$}
\end{center}

where
 the $L_{j}, R_{j}$ are defined inductively as follows:
\begin{align*}
L_{1}& =g_{1}(\vec m, k, \omega)  & R_{1}&= 0 \\
L_{j+1}& = g_{j+1}(\vec m, k+ R_{j+1} , \omega)  & R_{j+1}& = R_{j}+\pi_{2}(L_{j})
\end{align*}

It can be checked that, by construction, $g(m_{1},\dots, m_{n}, k,\omega)=\langle h_{m_{1},\dots, m_{n}}(v), |v|\rangle$, where $v\in \mathsf{EXT}^{k}_{v'}$ and $v'$ uniquely decomposes as $s\cdot t_{1}\cdot \dots \cdot t_{p}$.
Using the equations above we thus have: 
\begin{align*}
f(m_{1},\dots, m_{n})(y) &= \sum_{ h_{m_{1},\dots, m_{n}}(v)=y} \frac{1}{2^{|v|}}
\end{align*}

\end{varitemize}


\end{proof}
%
%
%
%

The desired Proposition \ref{prop2} can now be proved as
a corollary of Lemma \ref{lemma:aux}.
Indeed, once it is observed that 
if $g$ is obtained from
$f\in$, as in Lemma \ref{lemma:aux}, and 
$h(m_{1}, \dots, m_{n}, \omega) = \pi_{1}(g(m_{1}, \dots, m_{n}, 0, \omega))$, then it follows that
$f = \qtorand{h}$.

}

\paragraph{The Proof of the Main Result.}
The last ingredient to establish Theorem~\ref{theorem:arithmetization} is the following lemma, easily proved by induction on the structure of $\QPR$ functions. 
\begin{lemma}\label{lemma:randomfunction}
For every oracle function $f \in \QPR$, the random function $\qtorand{f}$ is \arithmetical.
\end{lemma}
\longv{
\begin{proof}
By Definition \ref{auxiliary}, given an arbitrary oracle function $f \in \QPR$, $f(x_{1}, \dots, x_{m}, \omega)$ = $y$, the corresponding $\qtorand{f}$ is defined as follows
$$
\qtorand{f} = \muCyl(f^{*}(x_{1}, \dots, x_{m}, y)) = \muCyl(\{\omega \ | \ f(x_{1}, \dots, x_{m}, \omega) = y\}).
$$ 
Lemma~\ref{lemma:randomfunction} states that such 
$\qtorand{f}$ is \arithmetical, 
i.e.~there is an $\PPA$ formula
$\fone_{\qtorand{f}}$, 
such that for every $n_{1}, \dots, n_{m}, l$, 
$$
\muCyl\big(\sem{\fone_{\qtorand{f}}(n_{1}, \dots, n_{m},l)}{}\big) = \qtorand{f}(n_{1}, \dots, n_{m})(l).
$$
The proof of the Lemma is by induction on the structure of oracle recursive functions.
Actually, the only case which is worth-considering is the one of \emph{query} functions. 
Indeed, all the other cases are obtained by trivial
generalizations of the standard proof by G\"odel~\cite{Godel}.

\begin{varitemize}
\item For each \emph{basic} oracle function, $f\in \QPR$, 
the corresponding random function, $\qtorand{f}$, is \arithmetical. There are four possible cases:

\par Oracle zero function. Let $f_{0}\in \QPR$ be the oracle zero function.
For Definition~\ref{auxiliary},
$f_{0}(x_{1},$ $\dots, x_{n}, \omega) = y$, with $y=0$, 
is such that the corresponding $\qtorand{f_{0}}$ 
is defined by the $\PPA$ formula,\footnote{Indeed,
 $\qtorand{f_{0}}$ = $\muCyl(\{\omega \ | \ f(x_{1}, \dots, x_{m}, \omega) = 0\}) = 1$ and 
 $F_{\qtorand{f_{0}}} : y=\zero$ and, since $y = \zero$ as $\sem{y}{} = 0$ and $\sem{\zero}{} = 0$, $\muCyl(F_{\qtorand{f_{0}}}) = 1$.}
 $$
F_{\qtorand{f_{0}}} : y=\zero.
$$
 
\par Oracle successor function. Let $f_{s}\in\QPR$
be the successor function.
For Definition~\ref{auxiliary},
 $f_{s}(x, \omega) = y$, with $y$ = $x$+1, is such that 
$\qtorand{f_{s}}$ is defined by the $\PPA$ formula, 
$$
F_{\qtorand{f_{s}}} : \suc{x} = y.
$$

\par Oracle projection function. 
Let $f_{\pi_{i}}\in \QPR$ be the projection function
Definition~\ref{auxiliary}, $f_{\pi_{i}}(x_{1},$ $\dots, x_{k}, \omega) = y$, with $y = x_{i},$ be the
is such that
$f_{\pi_{i}}($ such that $\qtorand{f_{\pi_{i}}}$ is defined by the $\PPA$ formula
$$
F_{\pi_{i}^{k}} : x_{i} = y.
$$

\par Oracle query function.
Let $q\in \QPR$ be the query function.
For Definition~\ref{auxiliary}, 
$q(x, \omega) = \omega(x)$,
is such that $\qtorand{f_{q}}$ is
defined by the $\PPA$ formula
$$
F_{q} : x = \atm{x}.
$$

\item For each oracle function $f \in \QPR$
 obtained composition, primitive recursion, or minimization from oracle recursive functions, 
 whose corresponding random functions are \arithmetical, the corresponding random function $\qtorand{f}$ is \arithmetical. 
The proofs for these three cases are very similar 
to the standard ones.
As an example, let us consider the case of
(simple) composition only.

Oracle (simple) composition.
Let $f$ be obtained by composition from 
$h$ and $g$, i.e.~$f(x_{1},$ $\dots, x_{n}, \omega)$ = 
$h(g(x_{1}, \dots, x_{n}, \omega), \omega)$ = $v$. 
It is possible show that $\qtorand{f}(x_{1}, \dots, x_{n})(v)$ is \arithmetical, which is, for every $x_{1}, \dots, x_{n}, v \in \Nat$, there is a $\PPA$ formula, 
$\fone_{\qtorand{f}}(x_{1}, \dots, x_{n}, v)$, such that 
$\muCyl(\sem{\fone_{\qtorand{f}}(x_{1}, \dots, x_{n}, v)}{})$ = $\qtorand{f}$.
Indeed, the desired formula is 
$$\fone_{\qtorand{f}} = \exists v\big(\fone_{\qtorand{h}}(v,y) \wedge \fone_{\qtorand{g}}(x_{1}, \dots, x_{n}, v)\big)
$$
where, by IH, $\muCyl(\sem{\fone_{\qtorand{h}}}{})$ = $\qtorand{h}$ and $\muCyl(\sem{\fone_{\qtorand{g}}}{})$ = $\qtorand{g}$.
\end{varitemize}
\end{proof}
}
\longv{As seen, s}\shortv{S}ince for both $\QPR$ and $\PPA$
the source of randomness consists 
in a denumerable amount of random bits, the proof of Lemma \ref{lemma:randomfunction} is easy, and follows the standard induction of~\cite{Godel}. 
Theorem~\ref{theorem:arithmetization} comes out as a corollary of the Lemma~\ref{lemma:randomfunction} above, together
with Proposition~\ref{proposition}: 
any computable random function is in $\PR$, by Proposition~\ref{proposition}, and each $\PR$ function is \arithmetical, by Lemma~\ref{lemma:randomfunction} and Proposition~\ref{prop2}.
\longv{Indeed, by Proposition~\ref{prop2}, for any $f \in \PR$ there is a $g \in \QPR$ such that $f = \qtorand{g}$ and, since $g \in \QPR$, by Lemma~\ref{lemma:randomfunction}, $\qtorand{g}$ (= $f$) is arithmetical.
}

\section{Realizability}\label{Realizability}

In this section we sketch an extension of realizability, a well-known computational interpretation of Peano Arithmetics, to $\PPA$. The theory of realizability\shortv{~\cite{Troelstra98}}\longv{~\cite{Troelstra98,Godel58,Kreisel57,SorUrz}}, which dates back to Kleene's 1945 paper~\cite{Kleene45}, 
provides a strong connection between logic, computability, and programming language theory.
The fundamental idea behind realizability
is that from every proof of an arithmetical formula in $\HA$ or equivalently (via the G\"odel-Gentzen translation) in $\PA$, one can extract a program, called the \emph{realizer} of the formula, which encodes the computational content of the proof. In Kreisel's \emph{modified-realizability}~\cite{Kreisel57}  realizers are typed programs: any formula $A$ of $\HA$ is associated with a type $A^{*}$ and any proof of $A$ yields a realizer of type $A^{*}$. 

Our goal is to show that the modified-realizability interpretation of $\HA$ can be extended to the language $\PPA$.
As we have not introduced a proof system for $\PPA$ yet, 
we limit ourselves to establishing the soundness of modified-realizability with respect to the semantics of $\PPA$.
Similarly to what happens with the class $\QPR$, 
the fundamental intuition is that realizers correspond to programs which can query an \emph{oracle} $\omega\in\Bool^{\Nat}$. For instance, a realizer of $\BOX^{t/s}A$ is a program which, for a randomly chosen oracle, yields a realizer of $A$ with probability at least $\sem{t}{\env}/\sem{s}{\env}$.

%
%

Our starting point is a PCF-style language with oracles. The types of this language are generated by 
basic types $\mathrm{nat}, \mathrm{bool}$ and the connectives $\to$ and $\times$.
%
We let $\mathrm O:= \mathrm{nat}\to\mathrm{bool}$ indicate the type of \emph{oracles}.
For any type $\sigma$, we let $[\sigma]$ (resp.~$[\sigma]_{\mathrm O}$) indicate the set of closed terms of type $\sigma$ (resp.~of terms of type $\sigma$ with a unique free variable $o$ of type $\mathrm O$).
Moreover, for all $i\in\{0,1\}$ (resp.~$n\in \mathbb N$), we indicate as $\overline i$ (resp.~$\overline n$) the associated normal form of type $\mathrm{bool}$ (resp.~$\mathrm{nat}$).
For all term $M$ and normal form $N$, we let $M\Downarrow N$ indicate that $M$ converges to $N$.
For any term $M\in [\sigma_{\mathrm O}]$ and oracle $\omega\in \Bool^{\Nat }$, we let $M^{\omega}\in[\sigma]$ indicate the closed program in which any call to the variable $o$ is answered by the oracle $\omega$.

We consider the language of $\PPA$ without negation and disjunction, enriched with implication $A\to B$.
As is usually done in modified-realizability, we take $\lnot A$ and $A\lor B$ as \emph{defined} connectives, given by $A\to (\mathtt 0=\mathtt S(\mathtt 0))$ and $\exists x. (x=\mathtt 0 \to A)\land(x=\mathtt S(\mathtt 0)\to B)$, respectively.
With any closed formula $A$ of $\PPA$ we associate a type $A^{*}$ defined as follows:\\
\medskip
\resizebox{0.95\textwidth}{!}{
\begin{minipage}{\linewidth}
\begin{align*}
\atm t^{*} & =  \mathrm{nat}    &   (\forall x.A)^{*} & = \mathrm{nat}\to A^{*} 
\\
(t=u)^{*} & =  \mathrm{bool}  &  (\exists x.A)^{*} &  = \mathrm{nat} \times A^{*}\\
(A\land B)^{*} & = A^{*}\times B^{*}   & (\BOX^{t/s}A)^{*} & =(\DIA^{t/s}A)^{*} 
 = \mathrm{O} \to A^{*}  \\ 
 (A\to B)^{*} &  = A^{*}\to B^{*} & & 
\end{align*}
\end{minipage}
}
\medskip

We define by induction the realizability relation 
$M, \omega \Vdash A$
 where $\omega\in \Bool^{\Nat}$ and, if $A=\BOX^{t/s}B$
 or $A=\DIA^{t/s}B$, $M\in [\sigma]$, and otherwise $M\in [\sigma]_{\mathrm O}$.

\begin{varenumerate}
\item $M, \omega \Vdash \atm t$ iff $\omega(\sem{t}{})=1$;
\item $M,\omega \Vdash t=s $ iff $\sem{t}{}=\sem{s}{}$;
%
\item $M, \omega \Vdash A_{1}\land A_{2}$ iff $\pi_{1}(M), \omega \Vdash A_{1}$ and $\pi_{2}(M),\omega\Vdash A_{2}$;

\item $M, \omega \Vdash A\to B$ iff $\omega\in \sem{A\to B}{}$ and 
$P,\omega\Vdash A$ implies $MP, \omega \Vdash B$;
\item $M, \omega \Vdash\exists x.A$ iff $\pi_{1}(M^{\omega})\Downarrow \overline k$ and $\pi_{2}(M), \omega\Vdash A(k/x)$;

\item $M,\omega \Vdash \forall x.A$ iff for all $k\in \mathbb N$, $M\overline k, \omega \Vdash A(k/x)$;
\item $M, \omega \Vdash \BOX^{t/s}A$ iff $\sem{s}{}> 0$ and 
$\muCyl\left(\{\omega'\mid Mo, \omega' \Vdash A\}\right)
\geq{ \sem{t}{}}/{\sem{s}{}}$;
\item $M, \omega \Vdash \DIA^{t/s}A$ iff $\omega\in \sem{\DIA^{t/s}A}{}$, and $\sem{s}{}=0$ or  
$\muCyl\left(\{\omega'\mid Mo, \omega' \Vdash A\}\right)
<{ \sem{t}{}}/{\sem{s}{}}$.
%


\end{varenumerate}
Condition 7.~is justified by the fact that for all term $M$ and formula $A$, the set $\{\omega\mid M,\omega \Vdash A\}$ can be shown to be measurable. Conditions 5.~and 9.~include a semantic condition of the form $\omega\in \sem{A}{}$, which has no computational meaning. This condition is added in view of Theorem \ref{thm:soundness} below. In fact, also in standard realizability
a similar semantic condition for implication is required to
show that realizable formulas are true in the standard model, see \cite{Troelstra98}.

\longv{

\begin{lemma}\label{lemma:owns}
For each term $M\in [\sigma]_{\mathrm O}$ and normal form $N\in [\sigma]$, 
the set $S_{M,N}$ of oracles $\omega$ such that $M^{\omega}\Downarrow N$ is measurable.
\end{lemma}
\begin{proof}
Let $\widehat M=\{M^{\omega}\mid \omega \in \Bool^{\Nat}\}$. 
Let $r$ be a reduction from some $P=M^{\omega}\in \widehat M$ to $N$; since $r$ is finite, it can only query finitely many values $b_{1},\dots, b_{n_{q}}$ of $\omega$.
Let  $C_{r}$ be the cylinder of all sequences which agree with $\omega$ at  $b_{1},\dots, b_{n_{q}}$. It is clear that for all $\omega'\in C_{r}$, $M^{\omega'}$ reduces to $N$. 
We have then that $S_{M,N}= \bigcup\{ C_{r}\mid \exists P\in \widehat M \text{ and }  r \text{ is a reduction from } P \text{ to }N \}$, so $S_{M,N}$ is a countable union of measurable sets, and it is thus measurable.
%
%
%
%
%
\end{proof}

\begin{lemma}\label{lemma:measu}
For each term $M$ and closed formula $A$, the set $S_{M,A}=\{\omega \mid M,\omega \vDash A\}$ is measurable.
\end{lemma}
\begin{proof}
By induction on the structure of $A$:
\begin{varenumerate}
\item if $A=\atm t$, then $S_{M,A}$ is the cylinder of all $\omega$ such that $\omega(\sem{t}{})=1$;

\item if $A=t=s$, then $S_{M,A}= \begin{cases}
\Bool^{\Nat}  & \text{ if } \sem{t}{}=\sem{s}{}\\
\emptyset & \text{ otherwise;}\end{cases}$


\item if $A=A_{1}\land A_{2}$, then $S_{M,A}= S_{\pi_{1}(M),A_{1}}\cap S_{\pi_{2}(M),A_{2}}$;

\item if $A=B\to C$, then
$S_{M,A}= \bigcap_{P\in [B^{*}]_{\mathrm O}}(S_{P,B}\cap S_{MP,C})$;

\item if $A=\exists x.B$, then $S_{M,A}= \bigcup_{k}( S_{\pi_{1}(M), \overline k}\cap
S_{\pi_{2}(M), A(k/x)})$;

\item if $A=\forall x.B$, then $S_{M,A}= \bigcap_{k}S_{M\overline k, A(k/x)}$;

\item if $A=\BOX^{t/s}B$, then $S_{M,A}=\begin{cases} \Bool^{\Nat} & \text{ if } \sem{s}{}>0 \text{ and }\muCyl(S_{Mo, B})\geq \sem{t}{}/\sem{s}{} \\ \emptyset & \text{ otherwise;}\end{cases}$
\item if $A=\DIA^{t/s}B$, then $S_{M,A}=\begin{cases} \Bool^{\Nat} & \text{ if } \sem{s}{}=0 \text{ or } \muCyl(S_{Mo, B})< \sem{t}{}/\sem{s}{} \\ \emptyset & \text{ otherwise.}\end{cases}$
\end{varenumerate}
\end{proof}}

%
%
%
%
%
%
%
%
%
%


\begin{theorem}[Soundness]\label{thm:soundness}
For a closed formula $A$, if $M,\omega\Vdash A$, then $\omega\in \sem{A}{}$.
\end{theorem}
\longv{
\begin{proof}
By induction on $A$:
\begin{varenumerate}
\item if $A=\atm t$ and $M,\omega \Vdash A$, then $\omega(\sem{t}{})=1$, so $\omega \in \sem{\atm t}{}$;

\item if $A=t=s$ and $M,\omega \Vdash A$, then $\sem{t}{}=\sem{s}{}$, so $\omega \in \Bool^{\Nat}= \sem{A}{}$;



\item if $A=A_{1}\land A_{2}$ and $M,\omega \Vdash A$, then from $\pi_{1}(M),\omega\vDash A_{1}$ and $\pi_{2}(M),\omega\vdash A_{2}$ we deduce 
$\omega\in \sem{A_{1}}{} \cap \sem{A_{2}}{}=\sem{A}{}$;


%
%
\item if $A=B\to C$ and $M,\omega \Vdash A$, then by definition $\omega\in \sem{A}{}$;
%

%

\item if $A=\exists x.B$ and $M,\omega\Vdash A$ then $\pi_{1}(M^{\omega})\Downarrow \overline k$ and 
$\pi_{2}(M), \omega \Vdash B(k/x)$, so by IH $\omega\in \sem{B(k/x)}{}\subseteq \sem{A}{}$;

\item if $A=\forall x.B$ and $M,\omega \Vdash A$ then for all $k \in \mathbb N$, $ M\overline k, \omega \Vdash A(k/x)$, so by IH $\omega \in \sem{A(k/x)}{}$ and we conclude $\omega \in \sem{A}{}= \bigcap_{k}\sem{A(k/x)}{}$; 

%

\item if $A=\BOX^{t/s}B$ and $M, \omega \Vdash A$ then $\sem{s}{}>0$ and 
$\muCyl(S)\geq \sem{t}{}/\sem{s}{}$, where $S=\{\omega' \mid Mo, \omega'\Vdash B\}$.
%
%
Then, 
by IH
$S\subseteq \sem{B}{}$, hence $\muCyl(\sem{B}{})\geq \muCyl(S)\geq \sem{t}{}/\sem{s}{}$. 
We conclude then that $\sem{A}{}=\Bool^{\Nat}$ and thus $\omega \in \sem{A}{}$;
%
%
%
\item if $A=\DIA^{t/s}B$ and $M, \omega \Vdash A$ then by definition $\omega \in\sem{\DIA^{t/s}B}{}$.
%
%
%

\end{varenumerate}
\end{proof}
}
For example, 
the term $
M= \lambda o. \mathrm{fix}\big (\lambda f x.(\mathrm{iszero}(ox))(f(x+1))\langle x,x\rangle\big )\overline 0
$ realizes the valid formula $\BOX^{1/1}\exists x.\atm x$. $M$ looks for the first value $k$ such that $o(k)=1$ and returns the pair $\langle \overline k, \overline k\rangle$. 
Similarly, the program $\lambda xoyz.o(y)$, which checks whether the $y$-th bit of $\omega$ is true, 
realizes the formula $\forall x.\BOX^{1/2^{x}}\forall_{y\leq x}\atm y$.
With the same intuition, one can imagine how a realizer $M$ of the formula $F_{\mathtt{IMT}}$ can be constructed: given inputs $x, o, y$, $M$ looks for the first $k$ such that the finite sequence $o(y+k),o(y+k+1),\dots, o(y+k + {\ell(n))}$ coincides with the string coded by $x$ (where this last check can be encoded by a program $\lambda w.P(x,o,y,z,w)$), and returns the pair $\langle \overline k, \lambda w. P(x,0,y,\overline k, w)\rangle$.

\longv{
\section{Related Works}\label{RelatedWorks}

To the best of the authors' knowledge, the term ``measure quantifier'' was first introduced by Morgenstern in 1979 
in order to formalize the idea that a formula $\fone(x)$ is true \emph{for almost all x}~\cite{Morgenstern}.\footnote{
Morgenstern's definition was inspired by the notion of generalized quantifier, which was ``introduced to specify that a given formula was true for ``many $x$'s"'' \cite[p.~103]{Morgenstern}.
Morgenstern defined a language $L_{\mu}$, 
obtained by adding the measure quantifier $Q_{\mu}$ 
to the standard first-order grammar and 
presented the the central notions of his logic
as follows (actually, other extended languages are considered in~\cite{Morgenstern}):
\begin{quote}
\textsc{Definition 2.1} A \emph{measure structure} $\mathcal{U}$ is a pair $\mathcal{U} = (\mathcal{U}, \mu^{\mathcal{U}})$, where $\mathcal{U}'$ is a first-order structure, card$|\mathcal{U}| = k$, a measurable cardinal, and $\mu^{\mathcal{U}}$ is a nontrivial $k$-additive measure on $|\mathcal{U'}|$ which satisfies the partition property. [...]
\par \textsc{Definition 2.2} Define a language $L_{\mu}$ to be a first-order language together with a quantifier $Q_{\mu}$, binding one free variable, where a measure $\mathcal{U} \vDash Q_{\mu}v_{0}\varphi(v_{0})$ iff $\{x \in |\mathcal{U}| \ | \ \mathcal{U'} \vDash \varphi[x]\} \in \mu^{\mathcal{U}}$.~\cite[pp.~103-104]{Morgenstern}
\end{quote}}
In the same years, similar quantifiers were investigated from a model-theoretic perspective by H. Friedman
(see~\cite{steinhorn} for a survey). 
More recently, Mio et al.~\cite{MichalewskiMio,MSM} investigated the possibility for such quantifiers 
to define extensions of MSO.
Generally speaking, all these works have been 
strongly inspired by the notion of generalized quantifiers, 
which already appeared in a seminal work by Mostowski~\cite{Mostowski}.\footnote{
Specifically, generalized quantifiers were first introduced by Mostowski, as ``operators which represent a natural generalization of the logical quantifiers"~\cite[p.~13]{Mostowski} and have then been 
extensively studied in the context of finite-model theory~\cite{Lindstrom,Kontinen}. Second-order generalized quantifiers have been recently defined as well~\cite{Andersson}.}
Nevertheless, the main source of inspiration for our treatment of measure quantifiers comes from computational complexity, namely from Wagner's counting operators on classes of languages~\cite{Wagner84,Wagner,Wagner86}.\footnote{For further details, see \cite{ADLP}, where the model theory and proof theory of an extension of propositional logic with counting quantifiers is studied (in particular, the logic $\mathsf{CPL}_{0}$ can be seen as a ``finitary'' fragment of $\PPA$).
} 
 \longv{

}On the other hand, there is an extensive amount of publications dealing with different forms of probabilistic reasoning (without references to arithmetic).
Most of the recent probability systems have been developed in the realm of modal logic, starting with the seminal (propositional) work by Nilsson~\cite{Nilsson86}. 
From the 1990s on, first-order probability logic and (axiomatic) proof systems, have been independently introduced by Bacchus~\cite{Bacchus90a,Bacchus90b,Bacchus} and Fagin, Halpern and Megiddo~\cite{FHM,FH94,Halpern90,Halpern03}.
\longv{
Remarkably, Bacchus defined \emph{probability terms}, by means of a modal operator $\mathsf{prob}$ 
computing the probability of certain events, 
and \emph{probability formulas}, which are equalities between probability terms and numbers.
 A similar first-order probability logic was introduced by Fagin, Halpern and Megiddo~\cite{FHM} (and later studied in~\cite{FH94,Halpern90,Halpern03}), in which probability spaces define the underlying models, and can be accessed through the so-called weight terms. 
}
Another class of probabilistic modal logics have been designed to model Markov chains and similar structure, for example in~\cite{KOZEN1981328,Hansson1994,LEHMANN1982165,lmcs:6054}.
However, \longv{once again, }no reference to arithmetic is present in these works.

From the 1950s on, the interest for probabilistic algorithms and models started spreading~\cite{LMSS,Davis61,Carlyle,Rabin63,Santos68,Gill74,Simon75}. Nowadays, random computation is pervasive in many area of computer science and, several formal models are available, such as probabilistic automata~\cite{Segala95}, both Markovian and oracle probabilistic Turing machines~\cite{Santos69,Santos71,Gill74,Gill77}, and probabilistic $\lambda$-calculi~\cite{SahebDjaromi,JonesPlotkin,DPHW,DLZ,EPT}. 
As seen, also a well-defined probabilistic recursion theory has been developed by~\cite{GDLZ,DalLagoZuppiroli}.
Our definition of the class $\QPR$ is guided by both by recent $\PR$~\cite{GDLZ,DalLagoZuppiroli} and by classical recursion theory~\cite{Godel,ChurchKleene,Kleene36,Kleene36b,Kleene36c,Kleene43,Turing,Peter}.\footnote{For further details on the history of the notion of recursion, see~\cite{Soare}.} 
Our definition of random arithmetical formulas and Theorem \ref{theorem:arithmetization} generalize the original results by G\"odel~\cite[pp.~63--65]{Godel92}.
Also our study of realizability is inspired by classical works. 
The functional or D-interpretation 
was first introduced by G\"odel in 1958 in order to prove the consistency of arithmetic~\cite{Godel58}.\footnote{Actually, G\"odel started conceiving the D-interpretation in the late 1930s~\cite{Soare,AvigadFeferman}.} 
The theory was further developed by Kreisel, 
who introduced the notion of modified-realizability~\cite{Kreisel57}, starting from Kleene's  realizability~\cite{Kleene45}.

}

\section{Conclusion}\label{Conclusion}

\shortv{\paragraph{Future and Ongoing Work.}}
This paper can be seen as ``a first exploration'' of $\PPA$, providing some preliminary results, but also leaving many problems and challenges open. 
The most compelling one is certainly that of defining a proof system for $\PPA$, perhaps inspired from realizability. 
Furthermore, our extension of $\mathsf{PA}$ is \emph{minimal} by design. In particular, we confined our presentation to a unique predicate variable, $\atm x$. Yet, it is possible to 
consider a more general language with countably many predicate variables $\mathsf{FLIP}_{a}(x)$, and suitably-\emph{named} quantifiers $\BOX^{t/s}_{a}$ and $\DIA^{t/s}_{a}$ (as in \cite{ADLP}). 
We leave the exploration of this more sophisticated syntax to future work.
%
Another intriguing line of work concerns the study of bounded versions of $\PPA$, which may suggest novel ways of capturing probabilistic complexity classes, different from those in the literature, e.g.~\cite{Jerabek2007}.

\shortv{
\paragraph{Related Work.}
To the best of the authors' knowledge, the term ``measure quantifier'' was first introduced in 1979 to formalize the idea that a formula $\fone(x)$ is true \emph{for almost all values of x}~\cite{Morgenstern}. In the same years, similar measure quantifiers were investigated from a model-theoretic perspective by H. Friedman
(see~\cite{steinhorn} for a survey). 
More recently, Mio et al.~\cite{MSM} studied the application of such quantifiers to define extensions of MSO. 
However, the main source of inspiration for our treatment of measure quantifiers comes from computational complexity, namely from Wagner's counting operators on classes of languages~\cite{Wagner}.\footnote{For further details, see \cite{ADLP}, where the model theory and proof theory of an extension of propositional logic with counting quantifiers is studied (in particular, the logic $\mathsf{CPL}_{0}$ can be seen as a ``finitary'' fragment of $\PPA$).
} 
On the other hand, an extensive literature exists on formal methods to describe probabilistic reasoning (without any reference to arithmetic), in particular in the realm of modal logic~\cite{Bacchus,FHM}.
\longv{
Starting from the seminal (propositional) work by Nilsson~\cite{Nilsson86}, from the 1990s on, first-order probability logic and (axiomatic) proof systems have been (independently) introduced by Bacchus~\cite{Bacchus} and Fagin, Halpern and Megiddo~\cite{FHM,Halpern90}.}
Universal modalities show affinities with
counting quantifiers. However, the latter
focusses on \emph{counting} satisfying valuations,
rather than on identifying (sets of) worlds. Moreover, classes of probabilistic modal logics have been designed to model Markov chains and similar structure, e.g. in~\cite{lmcs:6054}.
\longv{
\cite{KOZEN1981328,Hansson1994,LEHMANN1982165,lmcs:6054}.
}
}

%
%
%
 \bibliographystyle{splncs04}
 \bibliography{main}






\end{document}